\documentclass[11pt,a4paper,reqno]{amsart}
\usepackage{amsmath}
\usepackage{amsfonts}
\usepackage[utf8]{inputenc}
\usepackage[english]{varioref}
\usepackage{dcolumn}
\usepackage[height=22cm , width = 16cm , top = 4cm , left = 3cm, a4paper]{geometry}
\usepackage[final]{graphicx}
\usepackage{epsfig}
\usepackage{epstopdf}
\usepackage{pstricks}
\usepackage{multirow}
\usepackage{psfrag}
\usepackage{booktabs}
\usepackage{delarray}
\usepackage{amssymb}
\usepackage{rotating}
\usepackage{subfigure}
\usepackage{layout}
\usepackage{amsthm}
\usepackage{amssymb,amsmath}

\usepackage{hyperref}
\usepackage{xcolor}
\usepackage{setspace}
\usepackage{colortbl}
\usepackage[format=hang]{caption}
\usepackage{caption}
\usepackage{epsfig}
\usepackage{epstopdf}
\usepackage{booktabs}
\usepackage{mathtools}
\usepackage[noadjust]{cite}
\DeclareUnicodeCharacter{00A0}{~}
\newtheorem{theorem}{Theorem}
{
	{
		{
			{
				\newcommand{\hs}{\quad}

				\newtheorem{thm}{Theorem}[section]

				\newtheorem{defin}[thm]{Definition}
				\numberwithin{equation}{section}
				
				\parindent 0pt

\title[Hybridized PROJECTED DIFFERENTIAL TRANSFORM METHOD FOR collisional-breakage equation ]{\textbf{Hybridized PROJECTED DIFFERENTIAL TRANSFORM METHOD FOR collisional-breakage equation }}
\author[Shweta, Saddam Hussain and Rajesh Kumar]{Shweta $^\dag$, Saddam Hussain$^{\dag,*}$, and  Rajesh Kumar$^\dag$}
\thanks{$^\dag$Department of Mathematics,
Birla Institute of Technology and Science, Pilani, Rajasthan-333031, India.\\
$^\ast$Corresponding author: Email: p20210073@pilani.bits-pilani.ac.in\\
Saddam Hussain: Email: p20200438@pilani.bits-pilani.ac.in\\
Rajesh Kumar: Email: rajesh.kumar@pilani.bits-pilani.ac.in.}
\begin{document}
\maketitle
\begin{quote}
{\textit{ Abstract: The non-linear collision induced fragmentation plays a crucial role in modeling several engineering and physical problems. In contrast to linear breakage, it has not been thoroughly investigated in the existing literature. This study introduces an innovative method that leverages the Elzaki integral transform as a preparatory step to enhance the accuracy and convergence of domain decomposition, used alongside the projected differential transform method to obtain closed-form or series approximations of solutions for the collisional breakage equation (CBE). A significant advantages of this technique is its capability to directly address both linear and nonlinear differential equations without the need for discretization or linearization. The mathematical framework is reinforced by a thorough convergence analysis, applying fixed point theory within an adequately defined Banach space. Additionally, error estimates for the approximated solutions are derived, offering more profound insights into the accuracy and dependability of the proposed method. The validity of this approach is demonstrated by comparing the obtained results with exact or finite volume approximated solutions considering several physical examples. Interestingly, the proposed algorithm yields accurate approximations for the number density functions as well as moments with fewer terms and maintains higher precision over extended time periods.}}
\end{quote}
\section{Introduction}
The particulate process is a dynamic phenomenon characterized by alterations in the physical properties of particles. These processes are extensively studied across various domains, including engineering, astronomy, chemistry, and physics. Within particulate processes, particles may undergo amalgamation to form larger entities or disintegration into smaller fragments, leading to changes in size, shape, volume, and other attributes. Among the numerous types of particulate processes, coagulation (aggregation), fragmentation (breakage), nucleation, and growth are the most prominent. The fragmentation process refers to the disintegration of larger clusters into diminutive fragments. The fragmentation event can be a desirable mechanism in certain processes such as milling, grinding, and protein filament division \cite{tournus2021insights, wang2021multiscale}. Fragmentation bifurcates into two distinct categories: linear and nonlinear fragmentation. Internal stresses or external forces primarily govern linear breakage. In contrast, nonlinear collisional breakage encompasses the complex dynamics of particle interactions, particularly collisional events that can lead to the formation of fragments exceeding the original particle volume \cite{martinez2010considerations, planchette2017colliding}. Non-linear breakage can transfer mass between impacting particles, resulting in bigger offspring particles than the parent particles. Thus, it encompasses a wider framework than linear breakage and significantly influence on several engineering activities and natural science, including protein filament division, pharmaceutical research, chemical engineering, astronomy, fluidized beds, raindrop formation, milling \cite{ chen2020collision, ma2020effects, fries2013collision, lombart2022fragmentation} and chaotic motion in turbulent flow \cite{chen2020collision}. The likelihood of these particles merging depends on several complex factors, including the energy lost during collisions, interactions influenced by fluid dynamics, and the properties of the particle surfaces.
\par
Despite its crucial role and significant contributions to numerous real-world applications, the mathematical literature on collisional breakage equations (CBE) remains relatively sparse within the scientific community's discourse. This research work examines nonlinear breakage, specifically pure binary collisional fragmentation, which occurs when particles collide inelastically and instantaneously. These mechanisms are important for understanding raindrop dispersion and communication networks, as well as engineering operations like fluidized beds and milling. In recent studies on bubble columns, several researchers have utilized a straightforward linear breakage model to describe the dynamics. However, they frequently overlook the high probability of bubble formation due to collisions, which leads to the development of complex mathematical model \cite{guo2017cfd, yang2017numerical}. As a result, investigating nonlinear models that account for collision-induced fragmentation is crucial for a more accurate depiction of bubble column dynamics. 
\subsection{Formulation of non-linear collisional breakage model}
In 1988, Cheng and Redner \cite{cheng1988scaling} provided a mathematical formulation of one-dimensional binary collisional breakage through an integro-differential equation as follows:
\begin{align}\label{Col}
	\frac{\partial w (n,\tau)}{\partial\tau}=	\int_{0}^{\infty}\int_{n}^{\infty}\mu(\epsilon,\rho) \alpha(n,\epsilon,\rho)w(\epsilon,\tau)w(\rho,\tau)d\epsilon d\rho-\int_{0}^{\infty}\mu(n,\epsilon)w(n,\tau)w(\epsilon,\tau)d\epsilon
\end{align}
supported with the initial distribution, $w(n, 0)=w_{0}(n).$\\
Here, $w(n, \tau)$ represents the particle number density function characterized by the internal properties $n$ at time $\tau\geq0$. The equation \eqref{Col} accounts for the temporal rate of change in the number density function. 
Term $\mu(n,\epsilon)$  quantifies the rate of collision among two separate particles of volumes $n$ and $\epsilon$, leading to the disintegration of at least one of the striking particles. The collision
between a particle pair of volumes $n$ and $\epsilon$ is equivalent to collision between the pair of volumes $\epsilon$ and $n$. Therefore, the kernel is assumed to be symmetric in its arguments, i.e.,
$\mu(n,\epsilon)=\mu(\epsilon,n)$ for all $n, \epsilon \in \mathbb{R}=(0,\infty)$ and $\tau>0.$\\
The breakage distribution function $\alpha(n,\epsilon,\rho)$ denotes the rate at which particles of volume $n$ are generated by the breakage of particles of volume $\epsilon$ resulting from collisions between $\epsilon$ and $\rho$. The breakage kernel $\alpha(n,\epsilon,\rho)$ satisfies the following two properties:
\begin{itemize}
\item[A.]   $\int_{0}^{\epsilon} \alpha(n,\epsilon,\rho)dn = N(\epsilon,\rho)<\infty$ for all $\epsilon,\rho >0 ;$ Number of daughter particles after breakage.
\item[B.]   $\int_{0}^{\epsilon} n \alpha(n,\epsilon,\rho)dn = \epsilon$ for all $(n,\epsilon) \in (0,\infty) \times (0,\infty); $  
This property is known as the mass conservation property. It ensures that the total volume of daughter particles produced as a result of the breakage process equals the volume of the mother particle.
\end{itemize}
The first term on the right-hand side of equation \eqref{Col}, known as the birth term, illustrates the rate at which particles with volume $n$ are generated when larger particle $\epsilon$ undergoes fragmentation due to collision with $\rho$. The second term, designated as the death term, elucidates the rate at which particles with volume $\epsilon$ are depleted due to interactions with other particles.\\
Furthermore, for deeper insights and further analysis of the system, we investigate integral properties such as moments, which are crucial in understanding various real-life phenomena due to their physical representation. The $k$-th order moment of a distribution function $w(n, \tau)$ is mathematically defined as:
\begin{align}\label{exactmoment}
	N_k(\tau) = \int_{0}^{\infty} n^k w(n, \tau )dn.
\end{align}
The zeroth order moment corresponds to the total number of particles within the system, representing the overall population of particles. The first moment describes the total mass of the system, while  the second moment ($k=2$) corresponds to total energy dissipated by the system during the process.
\subsection{State-of-the-art and its advancement}
Several studies examined numerical schemes and semi-analytical techniques (SAT) to investigate linear fragmentation and nonlinear coagulation models, see \cite{ hussain2023semi, arora2023comparison, kaushik2023novel, kaushik2023laplace} and further citations. The study in \cite{ziff1991new} offers an analytical solution to the linear fragmentation equation for certain kernel functions. The findings in \cite{peterson1986similarity, breschi2017note} explores similarity solutions, whereas \cite{hosseininia2006numerical} delves into the numerical solution of the model via adopting discrete element method. Extensive work has also been conducted on the well-posedness of coagulation in the context of linear breakage equations. The non-linear fragmentation equation has sparked growing attention in the scientific community. As model is an integro-partial differential equation with various parameters, including the breakage distribution function, collisional kernel, and number density function, solving this equation analytically for complex physical parameters poses a significant challenge. Further, the literature on collisional breakage predominantly focuses on the theoretical aspects of solution existence and uniqueness \cite{laurenccot2001discrete, giri2021existence, barik2020global}. In \cite{laurenccot2001discrete, walker2002coalescence}, the authors investigated the existence, uniqueness, and density conservation of solutions to the CBE. Additionally, \cite{giri2021existence} demonstrates the existence of weak solutions for the CBE with a broad range of unbounded collision kernels characterized by $\mu(\epsilon,\rho)= \epsilon^{i}\rho^{j}+ \rho^{i} \epsilon^{j} $ where $i \leq j \leq 1$.
In \cite{barik2020global}, the authors detailed the mechanisms underlying global classical solutions for aggregation and CBEs, specifically with a collision kernel that permits infinite growth at high volumes. Recently, numerical studies are done on CBE via finite volume methods \cite{das2024improved, das2020approximate}.

For an in-depth discussion of the difficulties associated with the numerical approaches for solving population balance models, readers are referred to \cite{singh2022challenges}. These difficulties likely underlie the development of various methods to find solutions.
Therefore, developing more reliable and versatile tools for solving differential equations remains an essential and active area of research. 
Unlike traditional methods, SATs do not rely on standard features such as linearization, mesh discretization, and predefined sets of basis functions. Over the past decade, researchers have extensively employed SATs to derive series approximation solutions for nonlinear partial and integro-partial differential equations \cite{ hussain2023semi, hussain2024analytical, noor2008homotopy }. 
\par This work presents a novel hybridized SAT for solving the model \eqref{Col}. Tarig Elzaki's development of the Elzaki transform effectively addresses linear and non-linear differential equations (DE) \cite{elzaki2011new, elzaki2011application}. The approcah is further enhanced by integrating it with the projected differential transform method \cite{jang2010solving}, that significantly improve accuracy and convergence in solving these problems \cite{akinfe2021solitary, akinfe2022implementation, ayaz2004applications}. The analytical flexibility of this approach facilitates the resolution of intricate issues and simplifies complex equations. In nutshell, the selection of the Elzaki projected differential transform method (EPDTM) is driven by its versatility, numerical efficiency, and potential to advance mathematical exploration in complex, real-world systems. In a related study, Hussain and Kumar \cite{hussain2024elzaki} used the EPDTM for solving multi-dimensional aggregation and combined aggregation breakage equations. This motivates us to extend the scheme to solve CBEs for different kernel parameters and initial conditions. By employing this algorithm, we aim to offer a rigorous and computationally efficient framework for simulating collisional breakage processes. Our research is designed to advance the development of reliable solution techniques for population balance modeling, thereby aiding in designing and optimizing industrial processes involving particle breakage. Additionally, we include convergence analysis and error estimations to validate the method's reliability. Further six numerical examples are taken into account to compare our findings for moments as well as number density with available precise solution or FVM results. 
\par The rest of the paper is structured as follows: Section 2 presents a brief overview of the Elzaki transformation and the projected differential transform method. Section 3 describes the methodology, including the approach for solving the CBE and their respective convergence results. Section 4 assesses the accuracy and efficiency by comparing the approximated EPDTM solutions with exact solutions, where available or with the finite volume solutions \cite{das2020approximate}. Finally, observations and discussions are summarized in Section 5.
\section{Preambles}
\subsection{ Elzaki Integral Transform (EIT)}	
It represents an improved version of the extensively employed Laplace integral transform, whose efficacy has enabled the resolution of equations that proved challenging for both the  Sumudu and Laplace integral transforms \cite{elzaki2012elzaki}. 
The EIT is represented as a semi-infinite integral form, see \cite{elzaki2012new}, as 
\begin{equation} \label{e2.1}
E[f(\tau)] = T(q) = q\int_{0}^{\infty}f(\tau)e^{-\frac{\tau}{q}}d\tau,\quad\tau>0,
\end{equation}
or, similarly
\begin{align*}
E[f(\tau)] = T(q) = q^2\int_{0}^{\infty}f(q\tau )e^{-\tau}d\tau,\quad\tau>0.
\end{align*}
The function $f(\tau)$ should belong to the set:
$$ A' = \lbrace f(\tau): \exists G, q_{1}, q_{2} > 0, |f(\tau)|< Ge^{\frac{|\tau|}{q_j}}, \text{if} \quad \tau \in(-1)^j \times[0,\infty), j=1,2\rbrace.$$
Here, $G$ is a finite real number, whereas $ q_{1}$ and $q_{2}$  might be finite or infinite.
The inverse of EIT is defined as	
\begin{equation} \label{2.2}
E^{-1}[T(q)] = \frac{1}{2\pi i}\int_{0}^{\infty} e^{\tau q} T\bigg(\frac{1}{q}\bigg) q dq.
\end{equation}
Following \cite{elzaki2012elzaki}, the EIT of spatial and time derivatives for a function can be easily calculated using the appropriate integration approach and are given as
\begin{align*}
E\bigg[\frac{\partial f}{\partial \tau}\bigg] = &\frac{1}{q} T(n,q) -q f(n,0),\quad
E\bigg[\frac{\partial f}{\partial n}\bigg] = \frac{\partial T(n,0)}{\partial n},
\end{align*}
\text{and}	$$E\bigg[\frac{\partial^{n} f}{\partial \tau^{q}}\bigg]=\frac{T(q)}{q^n}-\sum_{k=0}^{n-1}v^{2-n+k}f^{k}(0). $$
By using \eqref{e2.1}, one can arrive at the following major properties of EIT:
\begin{align*}
E(1) =& q^2,\quad E(\tau^{n})= {n !} q^{n+2},\quad E[e^{a\tau}] = q\int_{0}^{\infty}f(\tau)e^{\frac{\tau}{q}}d\tau=\frac{q^{2}}{1-aq},\quad aq < 1,
\end{align*}
$$ E[f(\tau)+l(\tau)]=E[f(\tau)]+E[l(\tau)].$$
\subsection{Projected Differential Transform Method (PDTM)}
Bongsoo Jang \cite{jang2010solving} introduced the PDTM for analyzing nonlinear DEs. It represents a refined and enhanced rendition of Zhou's differential transform method \cite{zhou1986differential}, and relies on the commonly employed Taylor series. On the other hand, it varies from the Taylor series in determining coefficients \cite{zhou1986differential, duffy2004transform}. The fundamental definitions and properties of the PDTM are outlined below:
\begin{defin}
Let $g(n,\tau)$ be a function defined in $\mathbb{R}^2$. The projected differential transform (PDT) $G(n,k)$ of $g(n,\tau)$ concerning variable $\tau$ at $\tau_0$ is delineated as follows: 
\begin{align}\label{PDTM}	G(n,k)=\frac{1}{k!}\left[\frac{\partial^k}{\partial t^k}g(n,\tau)\right]_{\tau=\tau_0},
\end{align}
\end{defin}
and the inverse PDT of \eqref{PDTM} with respect to $\tau$ at $\tau_0$ expressed as follows:
\begin{align}\label{IPDTM}
	g(n,\tau)=\sum_{k=0}^{\infty}G(n,k)(\tau-\tau_0)^k.
\end{align}
As observed, equation \eqref{IPDTM} depicts the Taylor series of $g(n,\tau)$ at $\tau=\tau_0$. By correlating equations \eqref{PDTM} and \eqref{IPDTM}, the function $g(n,\tau)$ can be represented as follows: 
\begin{align}
	g(n,\tau)=\sum_{k=0}^{\infty}\frac{1}{k!}\left[\frac{\partial^k}{\partial t^k}g(n,\tau)\right]_{\tau=\tau_0}(\tau-\tau_0)^k.
\end{align} 
Based on the above definition and basic calculations, we have the following theorem.
\begin{theorem}
Consider $G(Y,k)$, $C(Y,k)$ and $J(Y,k)$ as the PDT of the functions $g(Y,\tau)$, $c(Y,\tau)$ and $j(Y,\tau)$, respectively, where $Y=(y_1,y_2,\ldots,y_n)$. Thus, the following  properties hold \cite{akinfe2021solitary}.
\end{theorem}
\begin{enumerate}
	\item{\textbf{Linear property for PDTM:}}\\ \\
	If $j(Y,\tau)$ = $\gamma g(Y,\tau)+\varphi c(Y,\tau)$,
	then $J(Y,k)$ = $\gamma G(Y,k)+\varphi  C(Y,k),$\\
	where $\gamma$ and $\varphi$ are arbitrary real constants.\\
	\item{\textbf{PDTM for multiple products:}}\\ \\
	Suppose we have three or more functions to be transformed such that:\\
	$$j(Y,\tau)=g_1(Y,\tau)\cdot g_2(Y,\tau)\cdot g_3(Y,\tau)\cdot \hdots \cdot  g_m(Y,\tau),  \text{ then} $$ 
	$$J(Y,k)=\sum_{k_{m-1}=0}^{k}\sum_{k_{m-2}=0}^{k_{m-1}}  \dots \sum_{k_1=0}^{k_2}G_1(Y,k_1)G_2(Y,k_2-k_1)\times \dots \times G_{m-1}(Y,k_{m-1}-k_{m-2})G_{m}(Y,k_{m}-k_{m-1}).$$
	\item{\textbf{PDTM for time derivative:}}\\ \\
	If $j(Y,\tau)=\frac{\partial^m}{\partial \tau^m}g(Y,\tau)$, $m = 1,2,3,\dots$ then $J(Y,k)=\frac{(k+m)!}{k!}G(Y,k+m)$.\\
	\item{\textbf{PDTM for space derivative:}}\\ \\
	If $j(Y,\tau)=\frac{\partial^m}{\partial y_i^m}g(y_1,y_2,y_3,\dots,y_m,\tau)$,\hs $i=1,2,\dots,m$ and\hs$m = 1,2,3,\dots$\\
	then  $J(Y,k)=\frac{\partial^m}{\partial y_i^m}G(y_1,y_2,y_3,\dots,y_m,k)$.
\end{enumerate}
\section{ EPDTM for collisional-breakage equation}
Let us derive the mathematical formulation of the EPDTM for solving  CBE \eqref{Col}. By employing the EIT on both sides of equation \eqref{Col} and leveraging the properties of the transformation, we obtain:					 
\begin{align*}
	\frac{1}{q}E[w(n,\tau)] -qw(n,0) =E\left(\int_{0}^{\infty}\int_{n}^{\infty}\mu(\epsilon,\rho) \alpha(n,\epsilon,\rho)f_1(w)d\epsilon d\rho-\int_{0}^{\infty}\mu(n,\epsilon)f_2(w)d\epsilon\right),
\end{align*}
	where \begin{align}\label{gfunction}
	f_1(w)= w(\epsilon,\tau)w(\rho,\tau) \hs \text{and} \hs f_2(w)= w(n,\tau)w(\epsilon,\tau),
\end{align} are the nonlinear parts.
Employing the inverse EIT to the aforementioned equation yields:	
\begin{align}\label{operator}
	w(n,\tau) =w(n,0)+ E^{-1}\bigg[qE\left[\int_{0}^{\infty}\int_{n}^{\infty}\mu(\epsilon,\rho) \alpha(n,\epsilon,\rho)f_1(w)d\epsilon d\rho-\int_{0}^{\infty}\mu(n,\epsilon)f_2(w)d\epsilon\right]\bigg].
\end{align}
In the subsequent stage, PDTM is now used  to decompose the nonlinear parts as
\begin{align}
	f_1(w)= \sum_{r=0}^{k}w_r(\epsilon,\tau)w_{k-r}(\rho,\tau) \hs \text{and} \hs f_2(w) = \sum_{r=0}^{k}w_r(n,\tau)w_{k-r}(\epsilon,\tau).
\end{align}
By following EPDTM, the solution can be written in series form as 
\begin{align}\label{series}
	w(n,\tau)=\sum_{k=0}^{\infty}w_k(n,\tau).
\end{align}	
To determine the coefficients of the series solution, substitute the series into \eqref{operator} and compare the coefficients. The iterations are as follows: 
\begin{align}\label{iter}
	w_{k+1}(n,\tau)=& E^{-1}\bigg[qE\bigg[\int_{0}^{\infty}\int_{n}^{\infty}\mu(\epsilon,\rho) \alpha(n,\epsilon,\rho)\bigg(\sum_{r=0}^{k}w_r(\epsilon,\tau)w_{k-r}(\rho,\tau)\bigg)d\epsilon d\rho \nonumber\\ &-\int_{0}^{\infty}\mu(n,\epsilon)\bigg(\sum_{r=0}^{k}w_r(n,\tau)w_{k-r}(\epsilon,\tau)\bigg)d\epsilon\bigg]\bigg].
\end{align}
Hence, the kth order approximated series solution is calculated by
\begin{align}\label{truncatedseries}
	\zeta_k(n,\tau):=\sum_{i=0}^{k}w_i(n,\tau).
\end{align}
\subsection{Convergence analysis:}
This section uses the fixed point theorem to specify convergence conditions for the approximated solutions towards the exact solution. Further, the maximum error bound in a Banach space is discussed.
Considering a Banach space $\mathbb{X}_{1}= \{ w : [0,T]\times[0,\infty)\rightarrow [0,\infty)\} $ with associated norm
\begin{align}\label{NORM}
\|w\|	= \sup_{\tau\in [0,\tau_{0}]} \int_{0}^{\infty}(c n)^{\beta}\left|w(n,\tau)  \right|dn < \infty,\hs \text{and} \hs  c, \beta >0.
\end{align}
Let us write equation \eqref{operator} in operator form as
\begin{align}\label{10}
w(n,\tau) = \widetilde{\mathcal{D}}[w],
\end{align}
where 
\begin{align}\label{11}
	\widetilde{\mathcal{D}}[w]= w(n,0)+E^{-1} \{q E[D[w]]\}
\end{align} 
and $D[w] $ is defined as 
$$ D[w] = \int_{0}^{\infty}\int_{n}^{\infty}\mu(\epsilon,\rho) \alpha(n,\epsilon,\rho)f_1(w)d\epsilon d\rho-\int_{0}^{\infty}\mu(n,\epsilon)f_2(w)d\epsilon.$$
\begin{thm}\label{T1}
	The operator $\widetilde{\mathcal{D}}$ is a contractive map, i.e., $\|\widetilde{\mathcal{D}}w-\widetilde{\mathcal{D}}w^{\ast}\| \leq \delta \|w-w^{\ast}\|$, $\forall$ $w,w^{\ast}\in \mathbb{X}_{1}$ if the following conditions hold
	\begin{itemize}
		\item[S1] $\alpha(n,\epsilon,\rho)=\frac{\sigma n^{j-1}}{\epsilon^j}$, $\sigma >0$ and j= 1,2,\ldots, such that $\int_{0}^{\epsilon}n \alpha(n,\epsilon,\rho)dn = \epsilon$ 
		\item[S2]  $\mu(n,\epsilon)\leq (n\epsilon)^{\beta}, \beta>0 \quad \text{for all} \quad n, \epsilon \in (0,\infty) $
		\item[S3] $ c $ is taken such that $\epsilon < c,$ and, 
		\item[S4] $ \delta= \tau_{0}^2\exp(2\tau_{0}P)[\|w_{0}\|+\frac{4\sigma P}{\beta}(1+\tau_{0}P)] <1 $ for suitable $\tau_{0}, \quad  \text{where} \quad P = \|w_{0}\|+\frac{\sigma}{\beta} T\|w_{0}\|_{1}^2.$
	\end{itemize}
\end{thm}
\begin{proof}
	To  establish the contraction mapping for $\widetilde{\mathcal{D}}$, let us split the proof in following steps.\\
	\textbf{Step 1:}
	To construct the contraction mapping of $D$, the equation \eqref{10} can be expressed in the following equivalent form
	\begin{align*}
		\frac{\partial }{\partial \tau}[w(n,\tau)\exp[R[n,\tau,w]]] = \exp[R[n,\tau,w]] \int_{0}^{\infty}\int_{n}^{\infty} \mu(\epsilon,\rho) \alpha(n,\epsilon,\rho)w(\epsilon,\tau)w(\rho,\tau)d\epsilon d\rho,
	\end{align*}
where $R[n,\tau,w]=\int_{0}^{\tau}\int_{0}^{\infty}\mu(n,\epsilon)w(\epsilon,s)d\epsilon ds$.
Further, an equivalent operator $\widetilde{D}$ is given by
\begin{align*}
		\widetilde{D}[w]=w(n,0)\exp[-R(n,\tau,w)]+\int_{0}^{\tau} \exp[R(n,s,w)-R(n,\tau,w)]\int_{0}^{\infty}\int_{n}^{\infty}  \mu(\epsilon,\rho) \alpha(n,\epsilon,\rho)w(\epsilon,s)w(\rho,s)d\epsilon d\rho ds .
\end{align*}
Let, $w, w^{\ast} \in \mathbb{X}_{1} $, then we have
	\begin{align}\label{2co}
		\widetilde{D}[w]-	\widetilde{D}[w^{\ast}]= & w_{0}(n) [\exp[-R(n,\tau,w)]-\exp[-R(n,\tau,w^{\ast})]]+\nonumber \\&\int_{0}^{\tau}\exp[R(n,s,w)-R(n,\tau,w)] \int_{0}^{\infty}\int_{n}^{\infty}\mu(\epsilon,\rho) \alpha(n,\epsilon,\rho)w(\epsilon,s)w(\rho,s)d\epsilon d\rho ds  \nonumber \\&-\int_{0}^{\tau}\exp[R(n,s,w^{\ast})-R(n,\tau,w^{\ast})]\int_{0}^{\infty}\int_{n}^{\infty}\mu(\epsilon,\rho) \alpha(n,\epsilon,\rho)w^{\ast}(\epsilon,s)w^{\ast}(\rho,s)d\epsilon d\rho ds \nonumber \\
		= & w_{0}(n) \mathcal{G}(n,0,\tau)+ \int_{0}^{\tau} \mathcal{G}(n,s,\tau) \int_{0}^{\infty}\int_{n}^{\infty} \mu(\epsilon,\rho) \alpha(n,\epsilon,\rho)w(\epsilon,s)w(\rho,s)d\epsilon d\rho ds \nonumber \\& + \int_{0}^{\tau} \exp[R(n,s,w^{\ast})-R(n,\tau,w^{\ast})]\bigg[ \int_{0}^{\infty} \int_{n}^{\infty} \mu(\epsilon,\rho) \alpha(n,\epsilon,\rho)w^{\ast}(\epsilon,s)[w(\rho,s)-w^{\ast}(\rho,s)]d\epsilon d\rho+ \nonumber \\& \int_{0}^{\infty} \int_{n}^{\infty} \mu(\epsilon,\rho) \alpha(n,\epsilon,\rho)w(\rho,s)[w(\epsilon,s)-w^{\ast}(\epsilon,s)]d\epsilon d\rho \bigg] ds.
	\end{align}
Here $$\mathcal{G}(n,s,\tau) =\exp[R(n,s,w)-R(n,\tau,w)]-\exp[R(n,s,w^{\ast})-R(n,\tau,w^{\ast})]. $$
Using condition $S2$, one can easily obtain
	\begin{align}\label{12}
		\mathcal{G}(n,s,\tau) \leq & \exp \bigg\{-\int_{s}^{\tau} \int_{0}^{\infty} \mu(n,\epsilon) w^{\ast}(\epsilon,\tau) d\epsilon  d\tau \bigg\} (\tau-s)\| w-w^{\ast}\|,\nonumber \\
		\leq &  (\tau-s) \exp\{ (\tau-s) \mathcal{B}\} \| w-w^{\ast}\|\leq \mathcal{A} \| w-w^{\ast}\|,
	\end{align} 
	where, $\mathcal{A}= \tau \exp\{\tau \mathcal{B}\}$ and $\mathcal{B} = \max\{\| w\| , \| w^{\ast}\|\}$.
 Let us define a set $\mathcal{S}=\{ w \in \mathbb{X}_{1} : \|w\|\leq 2 P\}$. It can be easily demonstrated that $\widetilde{D}$ maps $\mathcal{S}$ to itself. Further, $w,w^{\ast}\in \mathcal{S} $, implies that $\mathcal{B} \leq 2P$. By taking norm on both sides of equation \eqref{2co} leads to
\begin{align*}
\|	\widetilde{D}[w]-\widetilde{D}[w]^{\ast}\|\leq& 
\mathcal{A} \| w_{0}\| \| w-w^{\ast}\| + \frac{\sigma}{\beta} \mathcal{A}  \| w-w^{\ast}\|\int_{0}^{\tau} \| w\|^2 ds + \frac{\sigma}{\beta} \int_{0}^{\tau} \exp\{ \tau \mathcal{B} \} [( \| w\| + \|  w ^{\ast}\| ) \|  w- w^{\ast}\|] ds \\
\leq & \mathcal{A} \bigg[  \|  w_{0}\| + \frac{\sigma}{\beta} \tau  \|  w\|^2 +  \frac{\sigma}{\beta} ( \|  w\| + \|  w ^{\ast}\| ) \bigg] \|  w- w^{\ast}\|\\
\leq & \Delta	\| w- w^{\ast}\|
	\end{align*}
with $\Delta = \tau_{0}\exp(2\tau_{0}P)[\| w_{0}\|+\frac{4\sigma P}{\beta}(1+\tau_{0}P)] <1 $ for suitable $\tau_{0}$. Hence, $\widetilde{D}$ is contractive.
The equivalency of $\widetilde{D}$ enables us to have that the operator $D$ is also contractive, i.e.,
	\begin{align}\label{cont_tN}
		\|D[ w]-D[ w^{\ast}]\| \leq \Delta \| w- w^{\ast}\|.
	\end{align}
\textbf{Step 2:}
By applying \eqref{cont_tN}, along with the definitions and fundamental properties of the Elzaki and inverse EITs, we can get 
\begin{align*}
		\|\tilde{\mathcal{D}} w-\tilde{\mathcal{D}} w^*\|&= \|E^{-1}[qE\left[D w\right]]-E^{-1}[qE\left[D w^{\ast}\right]]\| \\
		&= \sup_{\tau\in [0,\tau_{0}]} \int_{0}^{\infty}[(c n)^{\beta}| \frac{1}{2 \pi} \int_{0}^{\infty}[\frac{1}{q^2}\int_{0}^{\infty} (D w-D w^*)e^{-q\tau}d\tau]e^{q\tau}q dq|]dn\\
		&=\frac{1}{2 \pi} \int_{0}^{\infty}\bigg(\frac{e^{q\tau}}{q}\int_{0}^{\infty}[e^{-q\tau}[\sup_{\tau\in [0,\tau_{0}]} \int_{0}^{\infty}(c n)^{\beta}|D w-D w^*|dn]]d\tau\bigg)dq\\
		&\leq\frac{1}{2 \pi} \int_{0}^{\infty}\bigg(\frac{1}{q} \int_{0}^{\infty}\Delta \| w- w^{\ast}\| e^{-q\tau}d\tau\bigg)e^{q\tau}dq\\
		&=\mathcal{L}^{-1}\bigg[\frac{1}{q}\left[\mathcal{L}\left[\Delta\|w-w^{\ast}\|\right]\right]\bigg] \\
	&=\Delta \tau_{0}\| w- w^*\|. 
	\end{align*}
Hence, $\tilde{\mathcal{D}}$ is contractive if $\delta=\Delta \tau_{0}<1.$ 
\end{proof}	
\begin{thm}\label{T2}
The truncated solution converges to  the exact solution $w$ if the non-linear operator  $\tilde{\mathcal{D}}$ holds the conditions of Theorem \ref{T1} and $\| w_1\| < \infty$, $\delta <1$ with the maximum error bound 
\begin{align}\label{errorbound1d} 
		\| w-\zeta_{k}\| \leq \dfrac{\delta^\xi}{1-\delta} \| w_1\|.
\end{align}
\end{thm}
\begin{proof}
Given that,
\begin{align*}
\zeta_{k} = &\sum_{i=0}^{k} w_{i}(n,\tau)
= w(n,0)+\sum_{i=0}^{k-1}E^{-1}\bigg[qE\bigg[\int_{0}^{\infty}\int_{n}^{\infty}\mu(\epsilon,\rho)\alpha(n,\epsilon,\rho)\bigg(\sum_{r=0}^{i}w_r(\epsilon,\tau)w_{i-r}(\rho,\tau)\bigg)d\epsilon d\rho\\ &-\int_{0}^{\infty}\mu(n,\epsilon)\bigg(\sum_{r=0}^{i} w_r(n,\tau) w_{i-r}(\epsilon,\tau)\bigg)d\epsilon\bigg]\bigg]\\	=&w(n,0)+E^{-1}\bigg[qE\bigg[\int_{0}^{\infty}\int_{n}^{\infty}\mu(\epsilon,\rho)\alpha(n,\epsilon,\rho)\bigg(\sum_{i=0}^{k-1}\sum_{r=0}^{i}w_r(\epsilon,\tau)w_{i-r}(\rho,\tau)\bigg)d\epsilon d\rho\\ &-\int_{0}^{\infty}\mu(n,\epsilon)\bigg(\sum_{i=0}^{k-1}\sum_{r=0}^{i}w_r(n,\tau)w_{i-r}(\epsilon,\tau)\bigg)d\epsilon\bigg]\bigg].
\end{align*}
We have $\sum_{i=0}^{k}\sum_{r=0}^{i}w_r(\epsilon,\tau)w_{i-r}(\rho,\tau)\leq f_1(\zeta_k)$ and $\sum_{i=0}^{k}\sum_{r=0}^{i}w_r(n,\tau)w_{i-r}(\epsilon,\tau) \leq f_2(\zeta_k)$, utilizing these in the above equation give 
\begin{align}
\zeta_k(n,\tau) \leq w_0(n)+E^{-1}\bigg[qE\bigg[\int_{0}^{\infty}\int_{n}^{\infty}\mu(\epsilon,\rho)\alpha(n,\epsilon,\rho)f_1(\zeta_{k-1})d\epsilon d\rho-\int_{0}^{\infty} \mu(n,\epsilon)f_2(\zeta_{k-1})d\epsilon\bigg]\bigg].
\end{align} 
This leads to the equivalent operator form as follows:
\begin{align}
\zeta_k \leq \widetilde{\mathcal{D}}[\zeta_{k-1}].
\end{align}
Owing to the property of contraction mapping of $\tilde{\mathcal{D}}$, it is evident that
\begin{align*}
\|	\zeta_{k+1}-\zeta_k\| \leq \delta \|\zeta_k-\zeta_{k-1}\|\leq \delta^2\|\zeta_{k-1}-\zeta_{k-2}\|\leq\cdots\leq \delta^k\|	\zeta_{1}-\zeta_0\|.
\end{align*}
Now, for all $k,\xi \in \mathbb{N}$ with $k>\xi$, the triangle inequality yields
\begin{align*}
\|\zeta_k-\zeta_\xi\| &\leq \|	\zeta_k-	\zeta_{k-1}\|+\|	\zeta_{k-1}-\zeta_{k-2}\|+\cdots+\|\zeta_{\xi+1}-	\zeta_\xi\| \leq (\delta^{k-1}+\delta^{k-2}+\cdots+\delta^{\xi})\|	\zeta_1-	\zeta_0\|\\ &= \dfrac{\delta^\xi (1-\delta^{k-\xi})}{1-\delta}\|w_1\| \leq \dfrac{\delta^\xi}{1-\delta}\|w_1\|.
\end{align*}
The expression above approaches to zero as $\xi\rightarrow \infty$, which implies the existence of a function $\zeta$ such that $\lim\limits_{k\rightarrow \infty}	\zeta_k=\zeta$, therefore,
\begin{align*}
w(n,\tau)=\sum_{i=0}^{\infty} w_i=\lim\limits_{k\rightarrow \infty}	\zeta_k=\zeta.
\end{align*}
By taking the limit as $k\rightarrow \infty$ while keeping $m$ fixed, the theoretical error bound given by \eqref{errorbound1d} is derived.
\end{proof}
\section{Numerical Results}
In this section, we utilize the EPDTM to address the non-linear CBE with varying combinations of breakage distribution, collision kernels, and initial distributions. To emphasize the originality of our scheme, we compare our solutions with existing analytical ones, when available, using error tables and graphs. For instances where exact solutions are elusive, we compare the results obtained using the EPDTM with those derived from the FVM, as discussed in \cite{das2020approximate}. Notably, the FVM is appropriate for this comparison because it is best-suited to solve such model due to automatic mass conservation. Also, it is the only computational approach validated and tested for the CBE. Concerning our EPDTM, interestingly, in one example, we successfully derive the $k$th-order general term which allows us to express the solution in a closed form. In other cases, although we could determine the complete structure of the solution, the integration of more complex functions prevented us from achieving a fully comprehensive solution. To address this, we identified the iterative solution's initial terms, thereby enabling finite term series solutions.\\

\textbf{Test case 1}
For the first case, we consider binary breakage kernel $\alpha(n,\epsilon,\rho)= \frac{2}{\epsilon}$ and product collisional kernel $\mu(n,\epsilon)= n\epsilon$ with an exponential initial distribution  $w(n,0) = e^{-n}$. In the field of aerosol science, the product kernel elucidates the greater susceptibility of larger aerosol particles to collisional occurrences and subsequent fragmentation, resulting in the generation of smaller aerosol particles. The analytical solution for the number density distribution is studied in \cite{ziff1985kinetics} and provided as $w(n,\tau)= (\tau+1)^2 e^{-(\tau+1)n}$.\\
Using the iterative scheme provided in equation \eqref{iter}, the approximated solutions are obtained as
\begin{align*}
	w_0(n,\tau) =& e^{-n}, \hs
	w_1(n,\tau) = \tau \left(-e^{-n}\right) (n-2),\hs 
	w_2(n,\tau) =\frac{1}{2} \tau^2 e^{-n} \left(n^2-4 n+2\right),\\
	w_3(n,\tau) =& -\frac{1}{6} \tau^3 e^{-n} n \left(n^2-6 n+6\right), \hs
	w_4(n,\tau) =\frac{1}{24} \tau^4 e^{-n} n^2 \left(n^2-8 n+12\right).
\end{align*}
Following the similar pattern, the general term is achieved as
$$w_{k}=\frac{(-1)^k}{k!}e^{-n}\tau^{k}(n^k-2k n^{k-1}+k(k-1)n^{k-2}), \quad \text{for} \quad k\geq 0$$
Thus, the value of EPDTM series solution becomes
\begin{align*}
 \zeta_k(n,\tau)&=\sum_{i=0}^{k}w_i(n,\tau)= \sum_{i=0}^{k}
\frac{(-1)^i}{i!}e^{-n}\tau^{i}(n^i-2i n^{i-1}+i(i-1)n^{i-2}).
\end{align*}
It is evident that the preceding series leads to the precise solution
$w(n,\tau)= (\tau+1)^2 e^{-(\tau+1)n} $ as $k \to \infty$.
Hence, the EPDTM provides exactly the analytical solution for this example.\\

\begin{figure}[htb!]
	\centering
	\subfigure[Exact and EPDTM number density ]{\includegraphics[width=0.35\textwidth]{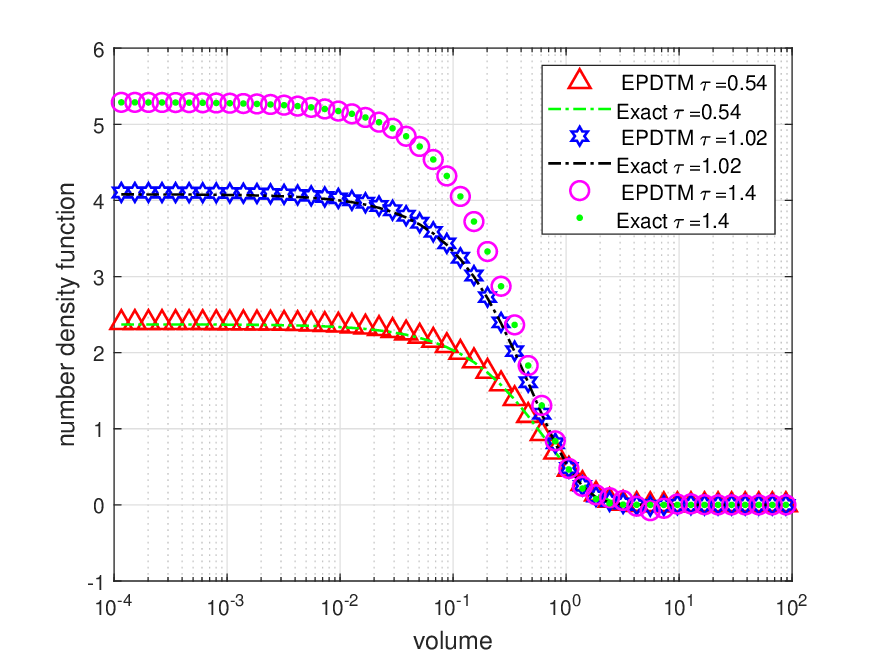}}
	\subfigure[Exact and EPDTM zeroth moments ]{\includegraphics[width=0.35\textwidth]{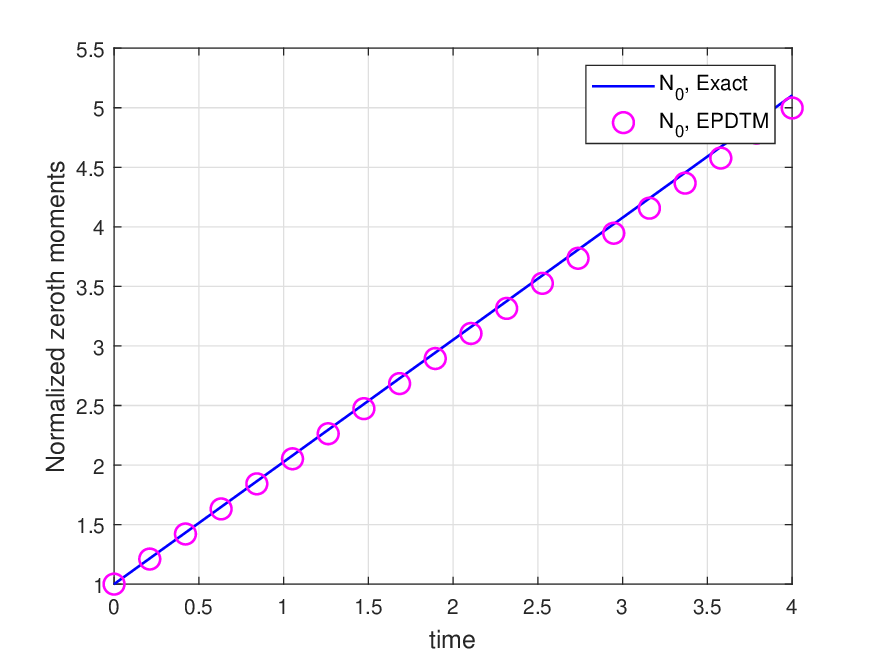}}
	\subfigure[Exact and EPDTM first moments ]{\includegraphics[width=0.35\textwidth]{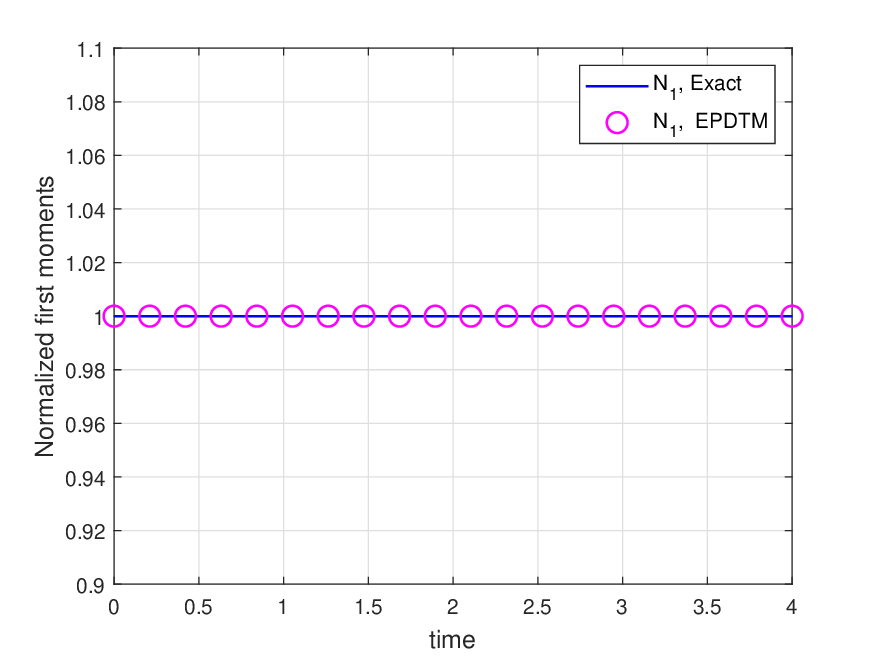}}
		\subfigure[Exact and EPDTM second moments ]{\includegraphics[width=0.35\textwidth]{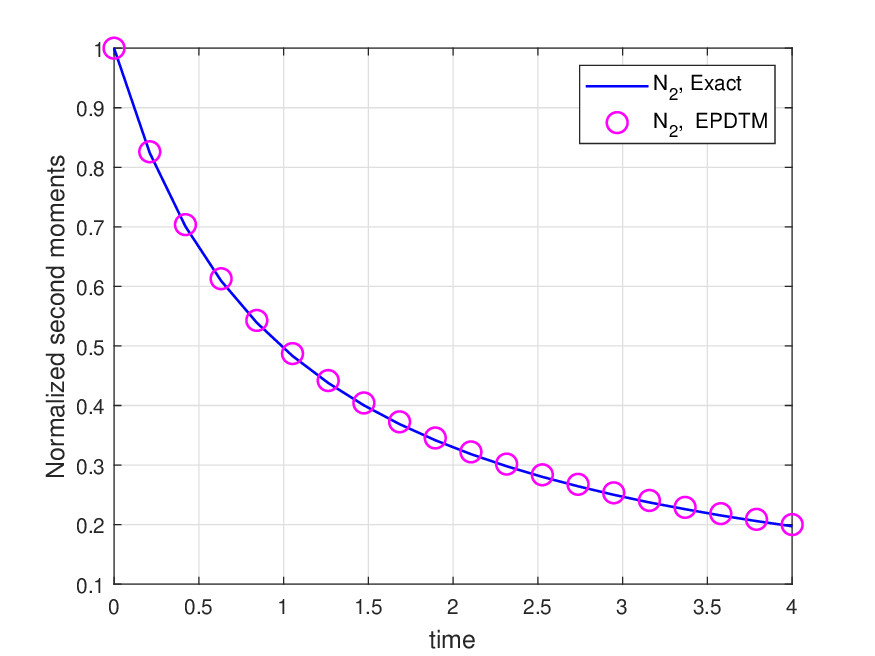}}
	\caption{Particle density distribution and moments functions for  Test case 1}
	\label{fig1}
\end{figure}
To evaluate its effectiveness, the $5^{th}$ iterative approximation for the number density function is compared with the exact solution  in Figure \ref{fig1}(a) at various time intervals. The figure illustrates that the EPDTM solution aligns perfectly with the analytical number density function. Initially, at time $\tau $= 0.54, the system contains fewer bubbles, which are also smaller in size. However, as time progresses to $\tau $= 1.02, there is a noticeable increase in the density of these small-sized bubbles. This increase is due to collisions among larger particles, causing them to fragment into smaller ones. When the system reaches $\tau $ = 1.4, it predominantly comprises smaller bubbles. This demonstrates that the EPDTM can accurately predict the behavior of bubbles within the system.
Table 1 provides the error between the exact and EPDTM approximated solutions, highlighting how discrepancies progress over time. Despite this, error can be controlled by taking more terms in the iterative series solution.  Additionally, Table 2 presents the numerical errors at computational time $\tau$ = 1 for various term solutions. Upon analyzing the errors and CPU processing times, it is clear that EPDTM offers a substantially better approximation.
Proceeding further, Figure \ref{fig1}(b, c, d) illustrates the zeroth, first and second moments of the density functions, respectively. The increasing trend of the zeroth moment in Figure \ref{fig1}(b) indicates that the total number of particles in the system rises over time, which is attributed to particle splitting following successful collisions. Despite this increase, the total mass of the particles remains conserved, as depicted in Figure \ref{fig1}(c). The decreasing trend of the second moment in Figure \ref{fig1}(d) shows that the system initially absorbs more energy and less over time. This occurs because the system initially contains larger particles, which require more energy to break apart. As tiny particles are produced through fragmentation, less energy is needed to break them. Interestingly, all the approximated moments calculated via the proposed algorithm overlap with the precise moments.
\begin{table}
\begin{center}
\begin{tabular}{ p{0.5cm}| p{3.5cm} p{3.5cm} p{3.5cm}  }
\multirow{2}{*} {$\tau$}
  
& Exact & EPDTM & EPDTM error\\ \hline
& & &\\ 
0.3 & 0.000690   & 0.000690 & 1.52217 E-5 \\ 
& & &\\ 
0.6 &0.0001734 & 0.000954 & 7.78138 E-4\\
& &   &\\ 
0.9 &0.00004042  &0.0073062 & 7.26578  E-3 \\
& &  &\\ 
1.2 & 8.95691 E-6  & 3.42759E-2 &  3.4267E-2 \\
& &  &\\ 
1.5 & 1.91189 E-6  & 1.1978E-1  & 1.1976E-1\\
& & & \\ 
\hline
\end{tabular}
\end{center}
\caption{Absolute errors for $n = 6$ at different time levels for Test case 1}
\label{table1}
\end{table}
\begin{table}
\begin{center}
\begin{tabular}{ |p{3.5cm}| p{3.5cm}| p{3.5cm}|   }
\hline
\multirow{2}{*} {No. of terms ($k$)} 
 
& Error & CPU time (sec.) \\ \hline
& & \\ 
5 & 3.4267 E-2   & 8.25  \\ 
& & \\ 
10 &2.98417E-3 & 12.88 \\
& &   \\ 
15 &9.36601E-4  & 15.40\\
& &  \\ 
20 & 2.0091 E-5  & 18.60  \\
& &  \\ 
\hline
		\end{tabular}
\end{center}
\caption{ Error and computational time for Test case 1 at $\tau$ = 1}
\label{table2}
\end{table}\\

\textbf{Test case 2}
Let us consider again the product coagulation kernel $\mu(n,\epsilon)= n\epsilon$  and the binary breakage kernel $\alpha(n,\epsilon,\rho)= \frac{2}{\epsilon}$ but with a Gaussian initial condition $w(n,0) = \frac{e^{\frac{-n^2}{2}}}{\sqrt{2 \pi}}$.
By employing the EPDTM given in \eqref{iter} for these set of parameters, we get
\begin{align*}
w_1(n,\tau) =&\frac{\tau e^{-\frac{n^2}{2}}}{2 \pi } \left(\sqrt{2 \pi } e^{\frac{n^2}{2}} \text{erfc}\left(\frac{n}{\sqrt{2}}\right)-n\right),\hs 
w_2(n,\tau) =-\frac{\tau^2 e^{-\frac{n^2}{2}}} {8 \pi ^{3/2}} \left(6n \sqrt{\pi } e^{\frac{n^2}{2}}  \text{erfc}\left(\frac{n}{\sqrt{2}}\right)-\sqrt{2} n^2-2 \sqrt{2}\right),\\
w_3(n,\tau) =& \frac{\tau^3 e^{-\frac{n^2}{2}} n} {24n \pi ^2} \left(6n \sqrt{2 \pi } e^{\frac{n^2}{2}}  \text{erfc}\left(\frac{n}{\sqrt{2}}\right)-n^2-6\right), \\	
w_4(n,\tau) =&-\frac{\tau^4 e^{-\frac{n^2}{2}} n^2} {192 \pi ^{5/2}} \bigg(20n \sqrt{\pi } e^{\frac{n^2}{2}} \text{erfc}\left(\frac{n}{\sqrt{2}}\right)-\sqrt{2} n^2-12 \sqrt{2}\bigg).
\end{align*}
Here, erfc($n$) denotes the complementary error function for non-negative values of $n$. 
\begin{figure}[htb!]
	\centering
	\subfigure[FVM and EPDTM number density ]{\includegraphics[width=0.35\textwidth]{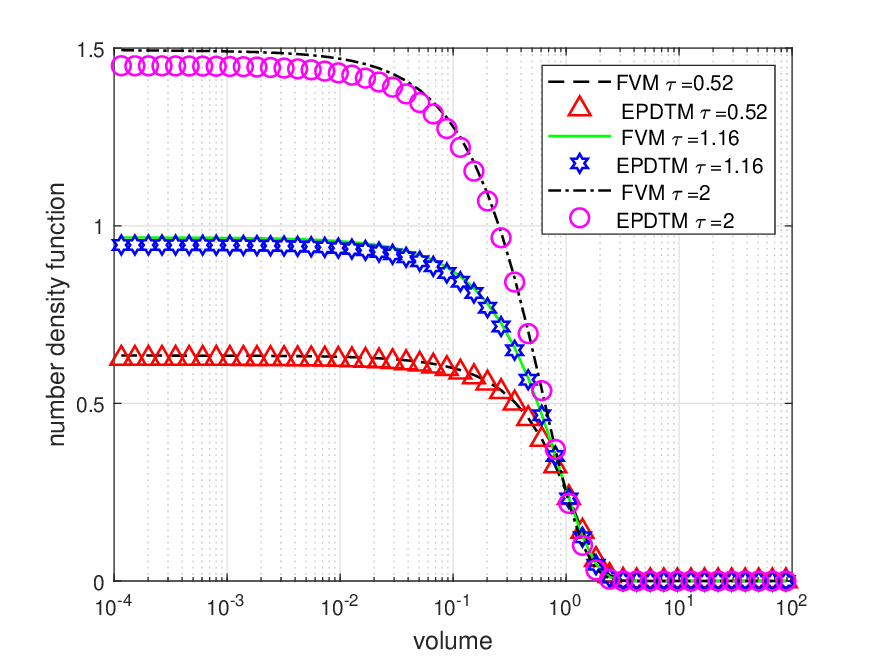}}
	\subfigure[FVM and EPDTM zeroth moments]{\includegraphics[width=0.35\textwidth]{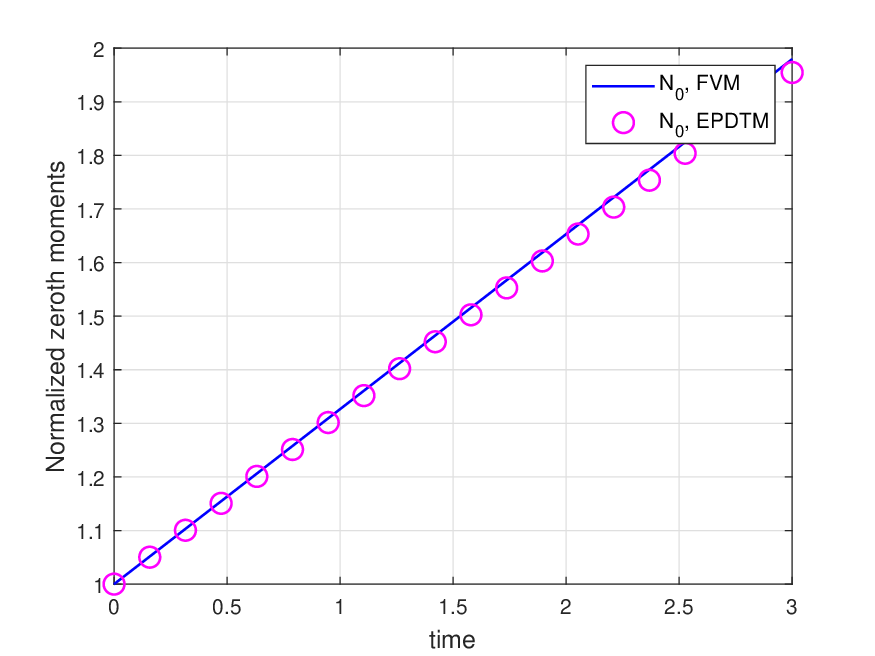}}
	\subfigure[FVM and EPDTM first moments ]{\includegraphics[width=0.35\textwidth]{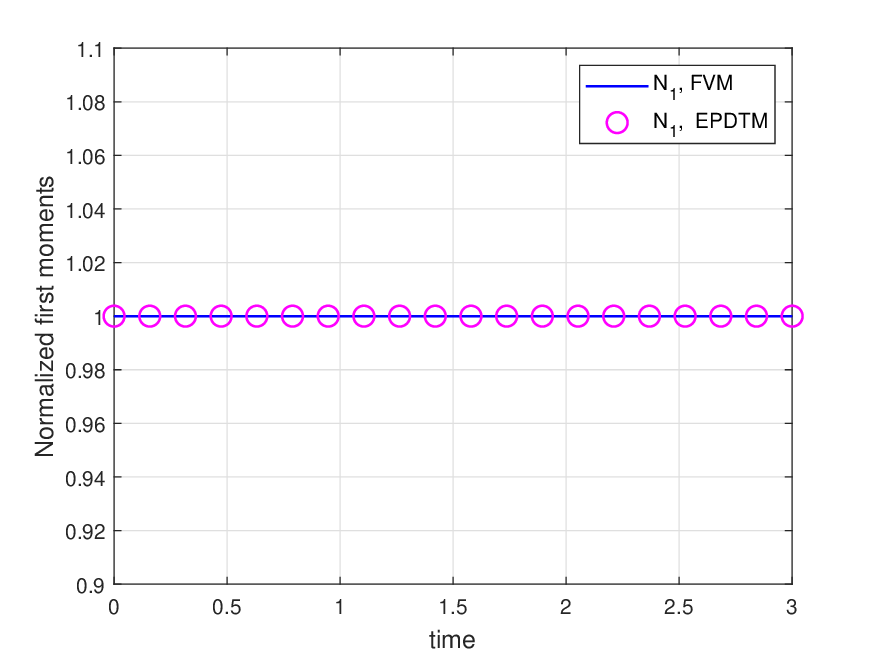}}
	\subfigure[FVM and EPDTM second moments ]{\includegraphics[width=0.35\textwidth]{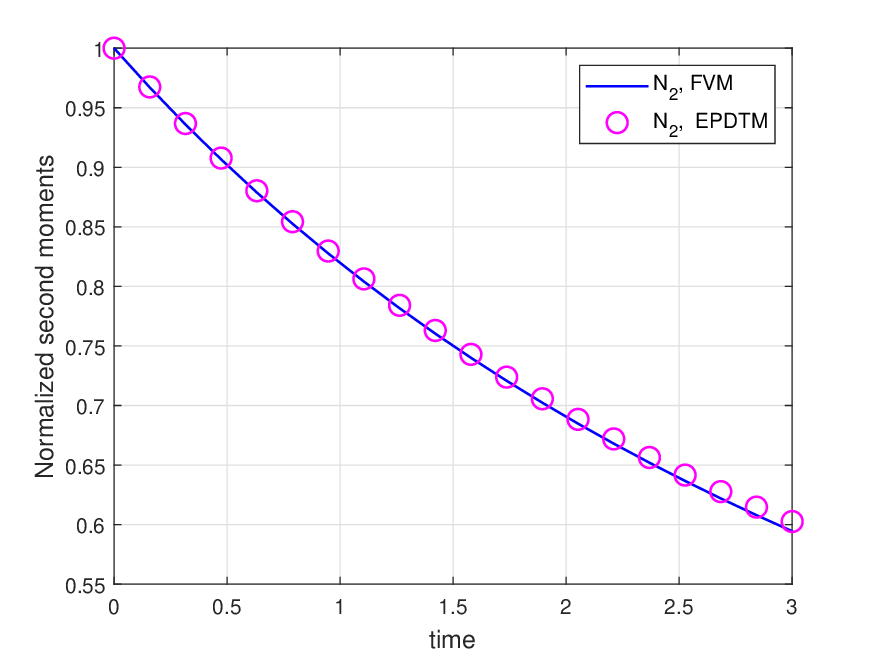}}
	\caption{Particle density distribution and moments functions for  Test case 2}
	\label{fig2}
\end{figure}
As the expression is getting complex, we were not able to find the closed form solution, though the higher order terms of the series can be easily calculated using MATHEMATICA. As previously stated, in the absence of an exact solution, we compare EPDTM’s six-term series approximation to FVM’s outcomes.
Figure \ref{fig2}(a) compares the number density functions computed using FVM and EPDTM at various time steps. The figure demonstrates that the EPDTM density function and the FVM density function are an outstanding complement. Moreover, Figure \ref{fig2}(b) displays the zeroth moment, showing linear growth for the chosen kernel parameters, with results almost overlapping. Figure \ref{fig2}(c) confirms the system’s mass conservation through the first moment. Additionally, the second moment derived using EPDTM closely matches the results from FVM, as illustrated in Figure  \ref{fig2}(d).\\

\textbf{Test case 3}
Now consider the polymerization kernel $\mu(n,\epsilon)= (n+a)^\frac{1}{3}(\epsilon+a)^\frac{1}{3}$ with $a = 0$ and binary breakage kernel $\alpha(n,\epsilon,\rho)= \frac{2}{\epsilon}$ with an exponential IC $w(n,0) = e^{-n}$.\\
The polymerization kernel is critical in predicting and controlling the molecular weight distribution, optimizing polymerization processes, and tailoring the properties of polymers for specific industrial applications. \\
Analogous to previous cases, the first few components of iterative series solutions are
\begin{align*}
	w_1(n,\tau) =&\tau e^{-n} \Gamma \left(\frac{4}{3}\right) \left(2 e^n \Gamma \left(\frac{1}{3},n\right)-\sqrt[3]{n}\right),\\ 
	w_2(n,\tau) =&-\frac{1}{6} \tau^2 e^{-n} \Gamma \left(\frac{4}{3}\right) \bigg(-3 n^{2/3} \Gamma \left(\frac{4}{3}\right)+\sqrt[3]{n} \Gamma \left(\frac{2}{3}\right)+42 e^n \sqrt[3]{n} \Gamma \left(\frac{4}{3}\right) \Gamma \left(\frac{1}{3},n\right)-3 e^n \Gamma \left(\frac{5}{3}\right) \Gamma \left(\frac{1}{3},n\right)\\&-30 e^n \Gamma \left(\frac{4}{3}\right) \Gamma \left(\frac{2}{3},n\right)\bigg).
\end{align*}
It can be seen that the approximated solution contains tedious terms, hence, a 3rd order approximated solution is considered for the  numerical simulations.		
\begin{figure}[htb!]
	\centering
	\subfigure[EPDTM and FVM number density ]{\includegraphics[width=0.35\textwidth]{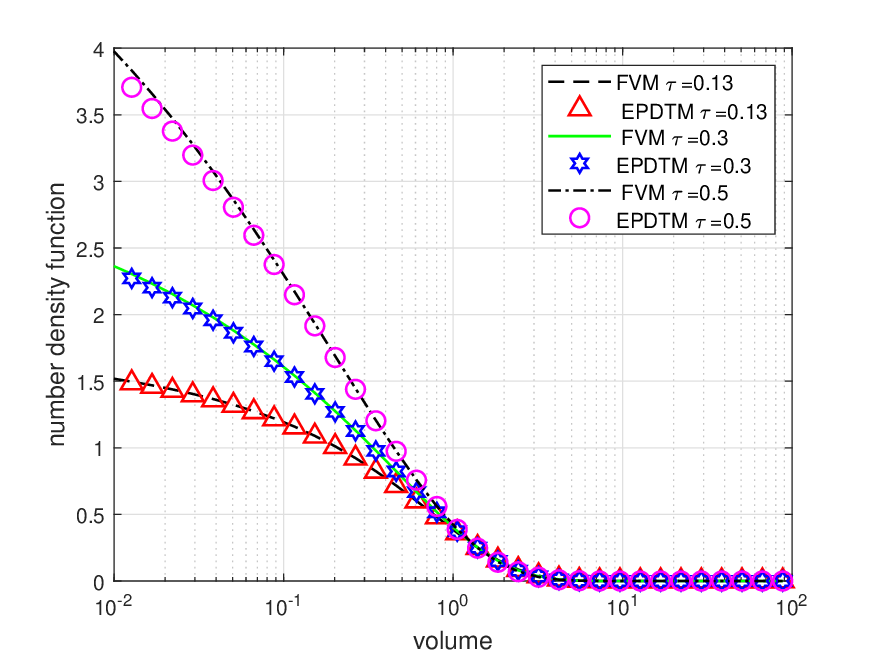}}
	\subfigure[FVM and EPDTM zeroth moments]{\includegraphics[width=0.35\textwidth]{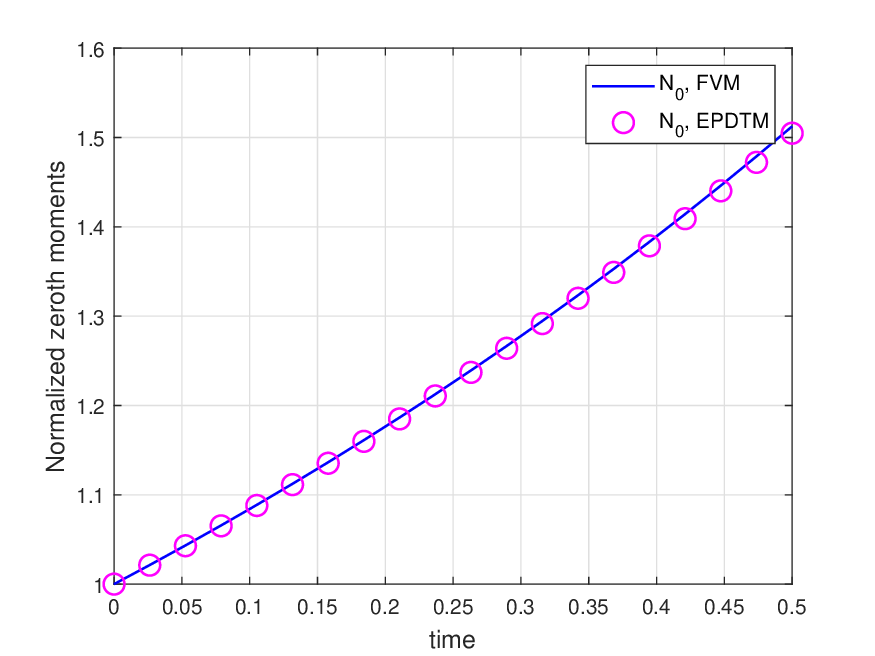}}
	\subfigure[FVM and EPDTM first moments]{\includegraphics[width=0.35\textwidth]{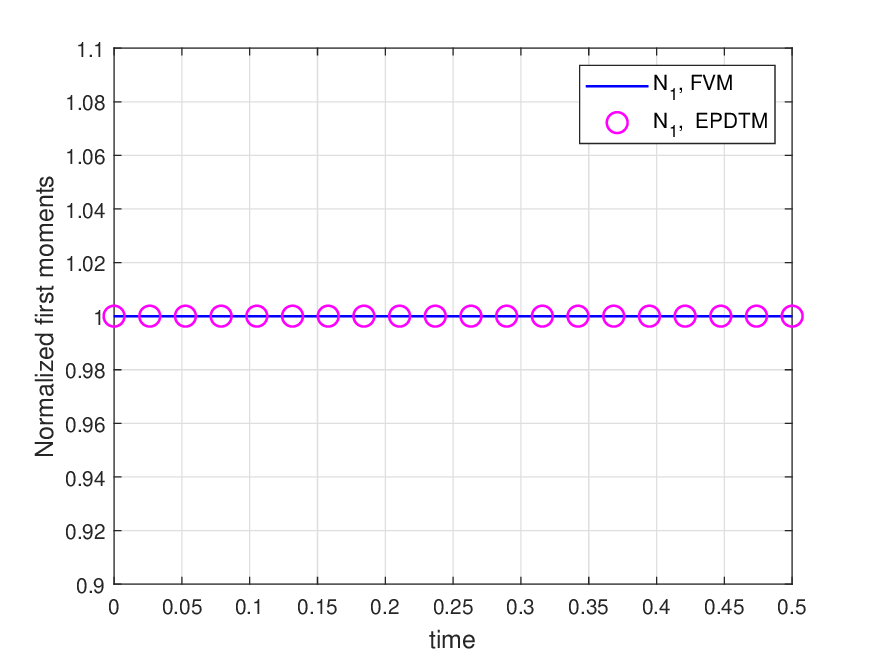}}
	\subfigure[FVM and EPDTM second moments]{\includegraphics[width=0.35\textwidth]{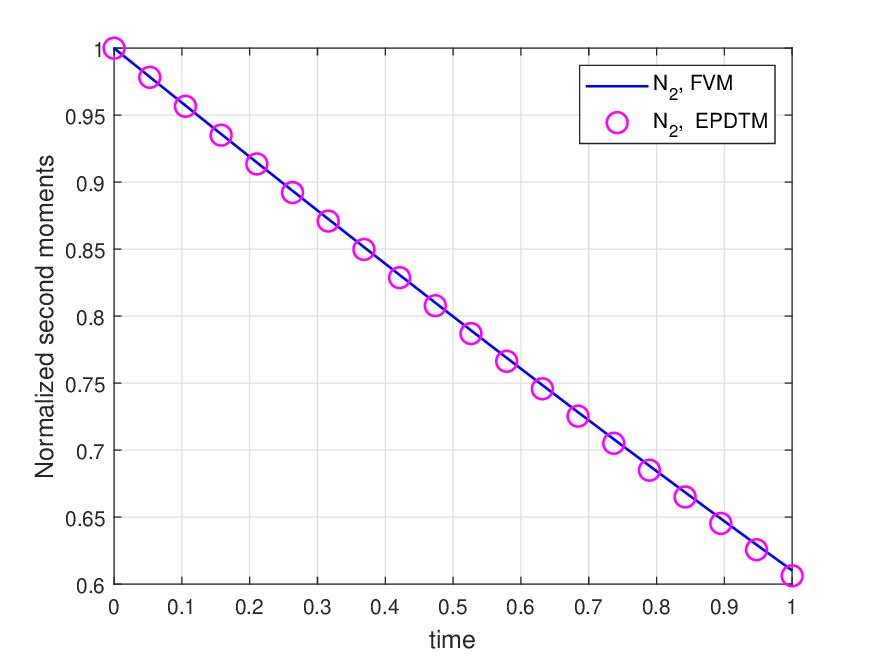}}
	\caption{Particle density distribution and moments functions for Test case 3}
	\label{fig3}
\end{figure}
Once again, due to the unavailability of analytical solutions, the newly obtained analytical results for the particle density functions and moments are compared against the existing FVM solutions to validate the accuracy of the proposed method. Figure \ref{fig3}(a) illustrates that the results for the particle density function closely match those obtained from FVM. From Figure \ref{fig3}(b), it is evident that the zeroth-order moment exhibits highly accurate results. The proposed scheme maintains adherence to the mass conservation law with just a few series terms (refer to Figure \ref{fig3}(c)). Also, the second moment obtained using EPDTM matches well with the FVM’s outcomes, see Figure \ref{fig3}(d).\\

\textbf{Test case 4}
To underscore the extensive versatility of the proposed method, we have opted for a more intricate kernel, i.e., Product coagulation kernel $\mu(n,\epsilon)= n\epsilon$, breakage kernel $\alpha(n,\epsilon,\rho)= \frac{3}{2 \epsilon^{\frac{1}{2}} n^{\frac{1}{2}}} $, wherein the fragmentation of a larger particle results in the formation of three smaller particles \cite{austin1976effect}. This kernel has been employed in investigations concerning the dynamics of ball milling and granulation processes \cite{park2020application, muanpaopong2023comparative}. \\
Following the recursive scheme defined in \eqref{iter} and using initial data $w(n,0) = e^{-n}$, the first few terms are defined as follows 
\begin{align*}
	w_1(n,\tau) =&\frac{\tau e^{-n} \left(3 \sqrt{\pi } e^n \text{erfc}\left(\sqrt{n}\right)-4 n^{3/2}+6 \sqrt{n}\right)}{4 \sqrt{n}},\\ 
	w_2(n,\tau) =&-\frac{\tau^2 e^{-n} \left(30 n \sqrt{\pi } e^n  \text{erfc}\left(\sqrt{n}\right)-9 \sqrt{\pi } e^n \text{erfc}\left(\sqrt{n}\right)+48 n^{3/2}-16 n^{5/2}-18 \sqrt{n}\right)}{32 \sqrt{n}},\\
	w_3(n,\tau) =& \frac{\tau^3 e^{-n}}{768 \sqrt{n}} \bigg(420 \sqrt{\pi } e^n n^2 \text{erfc}\left(\sqrt{n}\right)-180 n \sqrt{\pi } e^n \text{erfc}\left(\sqrt{n}\right)-45 \sqrt{\pi } e^n \text{erfc}\left(\sqrt{n}\right)-420 n^{3/2}+576 n^{5/2}\\&-128 n^{7/2}-90 \sqrt{n}\bigg).
\end{align*}
\begin{figure}[htb!]
\centering
\subfigure[EPDTM and FVM number density ]{\includegraphics[width=0.35\textwidth]{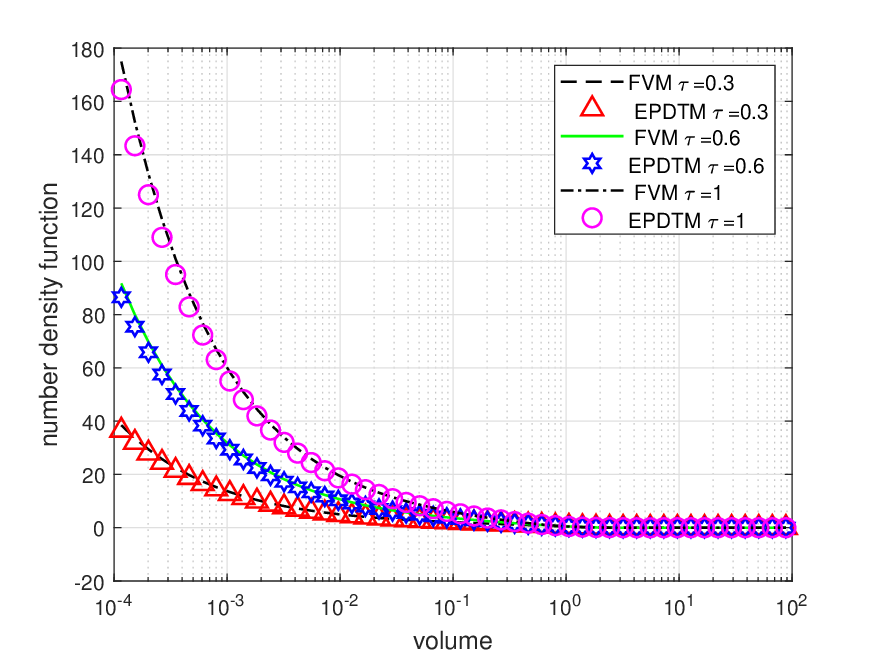}}
\subfigure[FVM and EPDTM zeroth moments]{\includegraphics[width=0.35\textwidth]{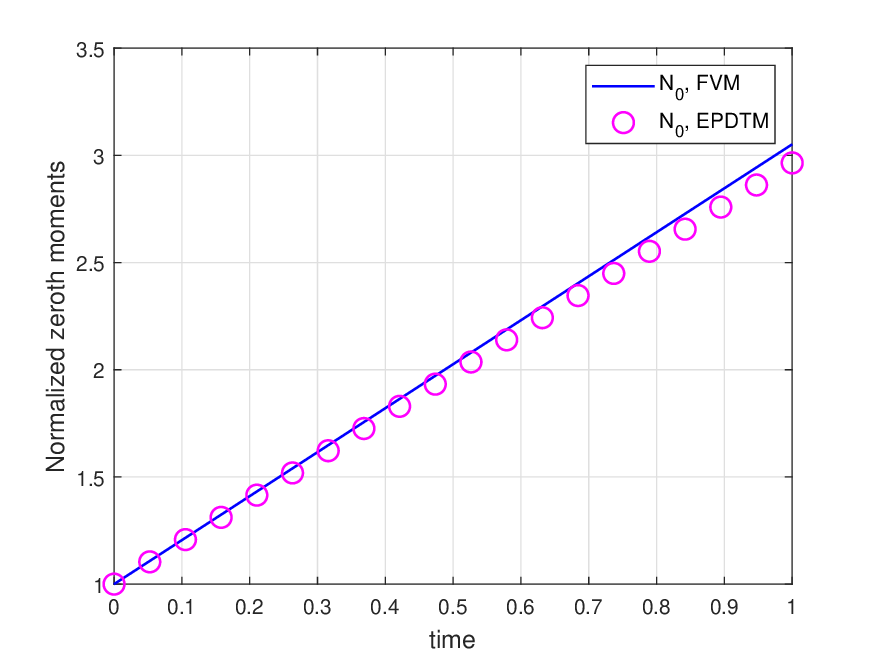}}
\subfigure[FVM and EPDTM first moments]{\includegraphics[width=0.35\textwidth]{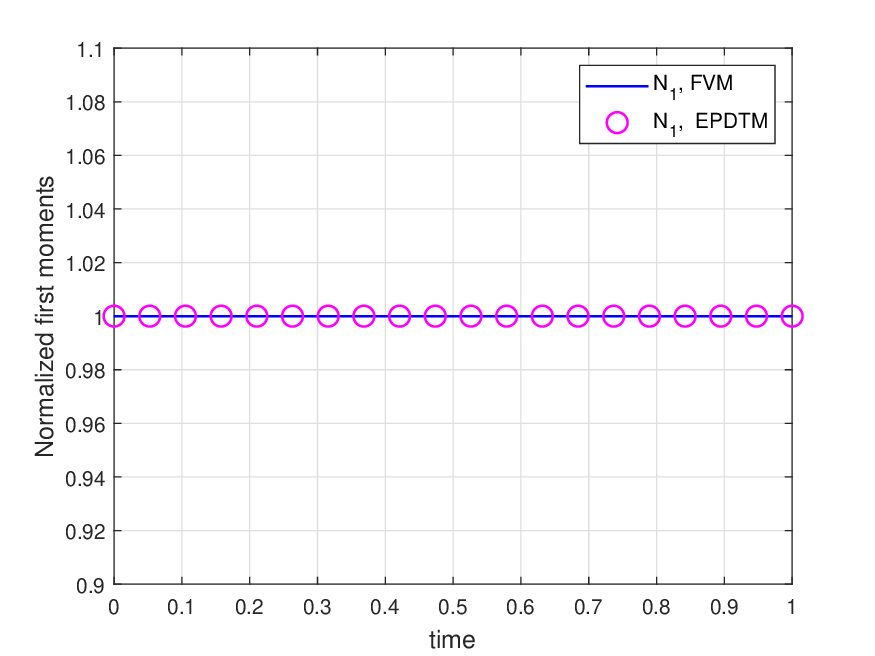}}
\subfigure[ FVM and EPDTM second moments]{\includegraphics[width=0.35\textwidth]{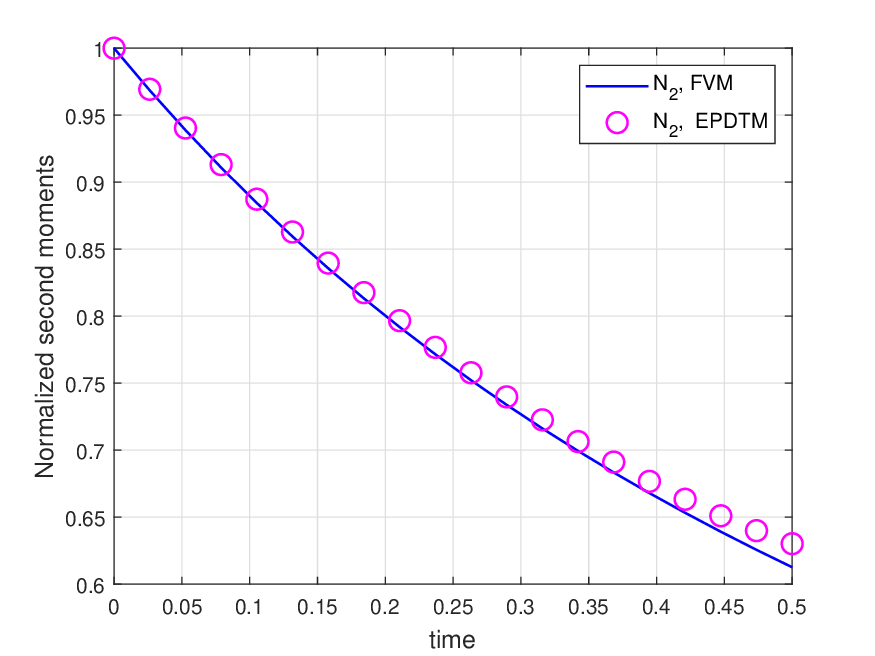}}
\caption{Particle density distribution and moments functions for Test case 4}
\label{fig4}
\end{figure}
By proceeding this way, one can derive higher order coefficients of the iterative series solution.
As mentioned earlier, the EPDTM number density functions align well with the FVM solutions, as illustrated in Figure \ref{fig4}(a). Furthermore, Figures \ref{fig4}(b) and \ref{fig4}(c) provide a qualitative comparison of the integral properties, i.e., zeroth and first-order moments, respectively. These figures demonstrate that the EPDTM computed moments closely resemble the FVM obtained moments. Once again, the proposed approach adheres to the mass conservation law and remains invariant for this complex kernel, as depicted in Figure \ref{fig4}(c). Further, it is visualize from Figure \ref{fig4}(d) that the EPDTM accurately captures the second moment for shorter time periods. As time progresses, discrepancies between the EPDTM and FVM second moments become apparent.\\

\textbf{Test case 5}
In this case, the proposed scheme is validated by using product coagulation kernel $\mu(n,\epsilon)= n\epsilon$ and ternary breakage kernel $\alpha(n,\epsilon,\rho)= \frac{3}{2 \epsilon^{\frac{1}{2}} n^{\frac{1}{2}}} $ with Gaussian initial condition $w(n,0) = \frac{e^{\frac{-n^2}{2}}}{\sqrt{2 \pi}}$.\\
Using the iterations defined in equation \eqref{iter}, we obtain
\begin{align*}
	w_1(n,\tau) =&\frac{\tau e^{-\frac{n^2}{2}} n} {8 \pi } \left(3 e^{\frac{n^2}{2}} E_{\frac{1}{4}}\left(\frac{n^2}{2}\right)-4\right),\hs 
	w_2(n,\tau) =\frac{\tau^2 e^{-\frac{n^2}{2}} n^2} {32 \sqrt{2} \pi ^{3/2}} \left(-15 e^{\frac{n^2}{2}} E_{\frac{1}{4}}\left(\frac{n^2}{2}\right)+3 e^{\frac{n^2}{2}} E_{-\frac{1}{4}}\left(\frac{n^2}{2}\right)+8\right),\\
	w_3(n,\tau) =&\frac{\tau^3 e^{-\frac{n^2}{2}} n^3} {768 \pi ^2} \left(70 e^{\frac{n^2}{2}} n^2 E_{-\frac{3}{4}}\left(\frac{n^2}{2}\right)-30 e^{\frac{n^2}{2}} E_{-\frac{1}{4}}\left(\frac{n^2}{2}\right)-3 e^{\frac{n^2}{2}} E_{-\frac{3}{4}}\left(\frac{n^2}{2}\right)-172\right).
\end{align*}
Here $E_{a}(b)$ refers to the exponential integral function.
\begin{figure}[htb!]
\centering
\subfigure[EPDTM and FVM number density ]{\includegraphics[width=0.35\textwidth]{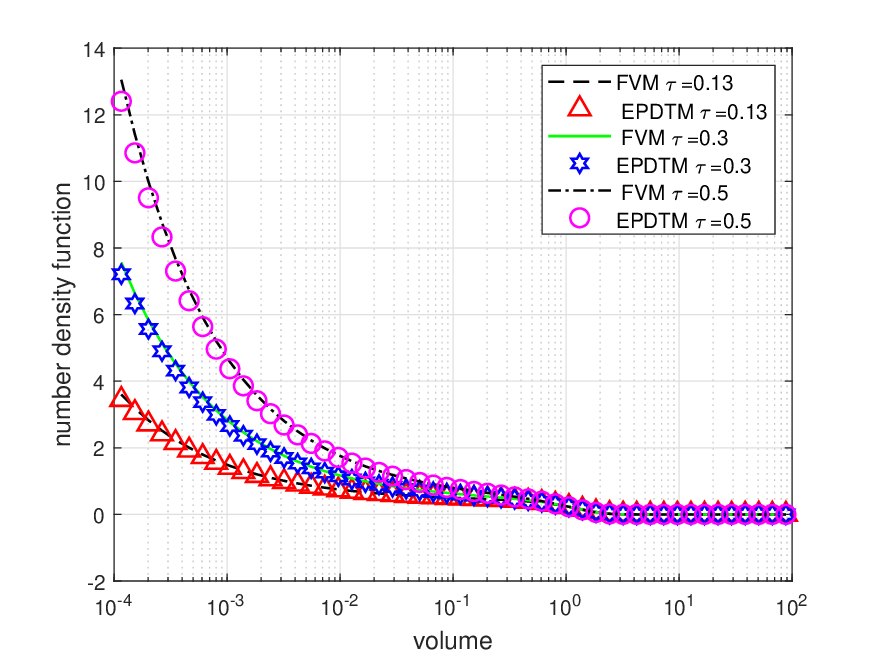}}
\subfigure[FVM and EPDTM zeroth moments]{\includegraphics[width=0.35\textwidth]{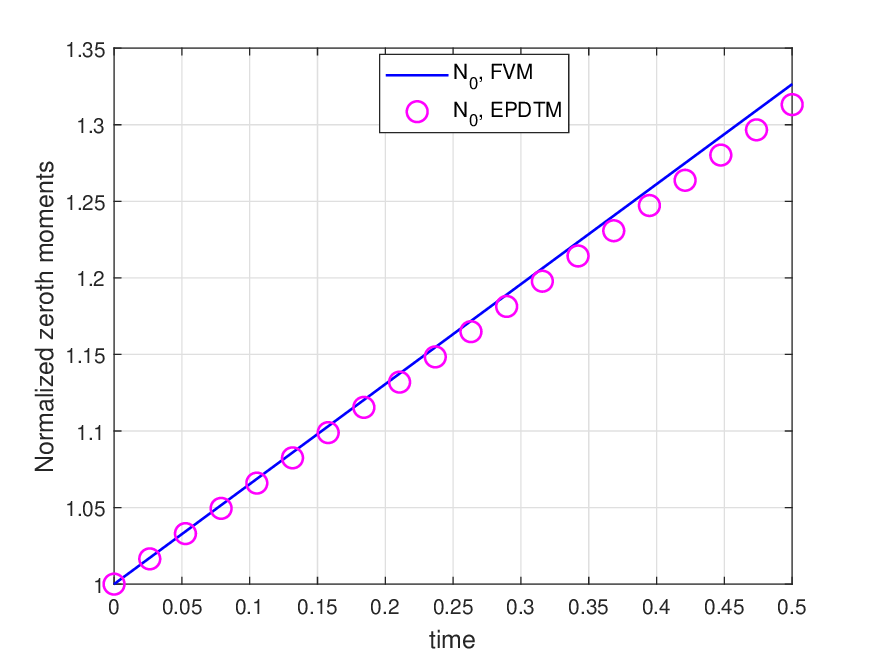}}
\subfigure[FVM and EPDTM first moments]{\includegraphics[width=0.35\textwidth]{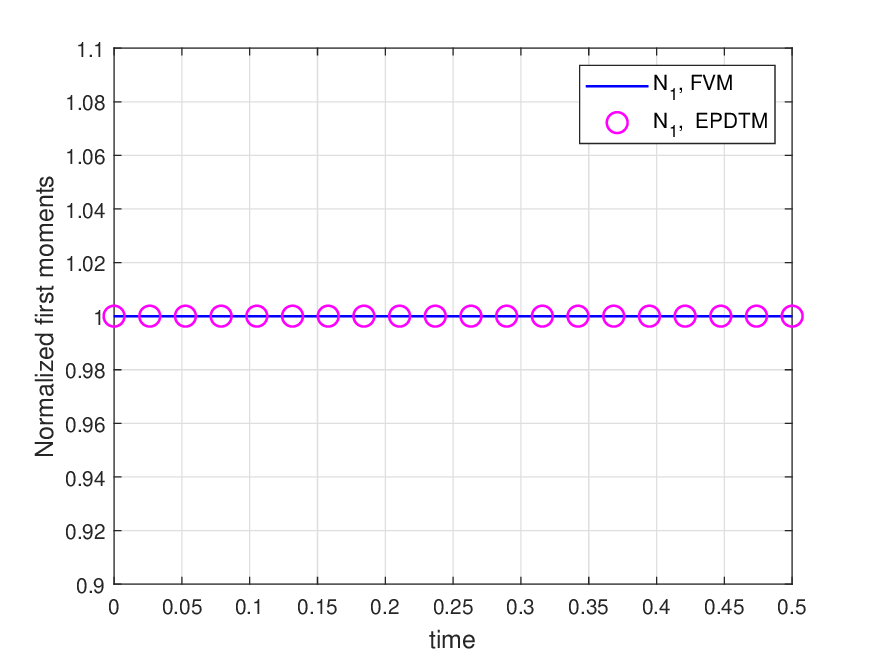}}
\subfigure[FVM and EPDTM second moments]{\includegraphics[width=0.35\textwidth]{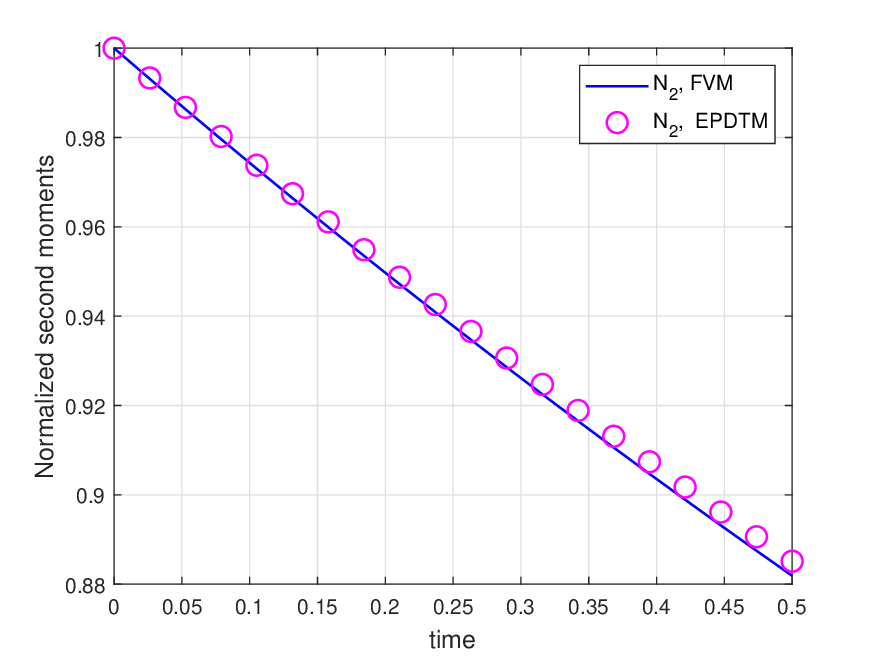}}
\caption{Particle density distribution and moments functions for Test case 5}
	\label{fig5}
\end{figure}
Owing to the intricate nature of the terms, a 5th order iterative series solution of the number density functions and their integral moments are compared with the FVM. Figure \ref{fig5}(a) demonstrates that  both the results obtained via FVM and EPDTM are overlapping to each other at various time levels. Furthermore, as shown in Figure \ref{fig5}(b, c, d), the integral moments determined using the EPDTM accurately produce the efficacy of the proposed scheme as they align well with the FVM results.\\ 

\textbf{Test case 6}
Now consider, product coagulation kernel $\mu(n,\epsilon)= \frac{n\epsilon}{20}$ and ternary breakage kernel $\alpha(n,\epsilon,\rho)= \frac{3}{2 \epsilon^{\frac{1}{2}} n^{\frac{1}{2}}}$ with exponential initial condition $w(n,0) =n^{2} e^{-n}$ .\\
Having the equation \eqref{iter}, we get the following iterations,
\begin{align*}
	w_1(n,\tau) =&\frac{3 \tau e^{-n}} {160 \sqrt{n}}\left(45 \sqrt{\pi } e^n \text{erfc}\left(\sqrt{n}\right)+60 n^{3/2}+24 n^{5/2}-16 n^{7/2}+90 \sqrt{n}\right),\\ 
	w_2(n,\tau) =&-\frac{9 \tau^2 e^{-n}}{12800 \sqrt{n}} \left(\sqrt{\pi } e^n  \text{erfc}\left(\sqrt{n}\right)(450 n-315)+480 n^{3/2}+432 n^{5/2}+192 n^{7/2}-64 n^{9/2}-630 \sqrt{n}\right),\\
	w_3(n,\tau) =&\frac{9 \tau^3 e^{-n}}{1024000 \sqrt{n}} \bigg(6300 \sqrt{\pi } e^n n^2 \text{erfc}\left(\sqrt{n}\right)-6300 \sqrt{\pi } e^n n \text{erfc}\left(\sqrt{n}\right)-2835 \sqrt{\pi } e^n \text{erfc}\left(\sqrt{n}\right)-16380 n^{3/2}\\&+2688 n^{5/2}+4608 n^{7/2}+2304 n^{9/2}-512 n^{11/2}-5670 \sqrt{n}\bigg).
\end{align*}
\begin{figure}[htb!]
\centering
\subfigure[EPDTM and FVM number density ]{\includegraphics[width=0.35\textwidth]{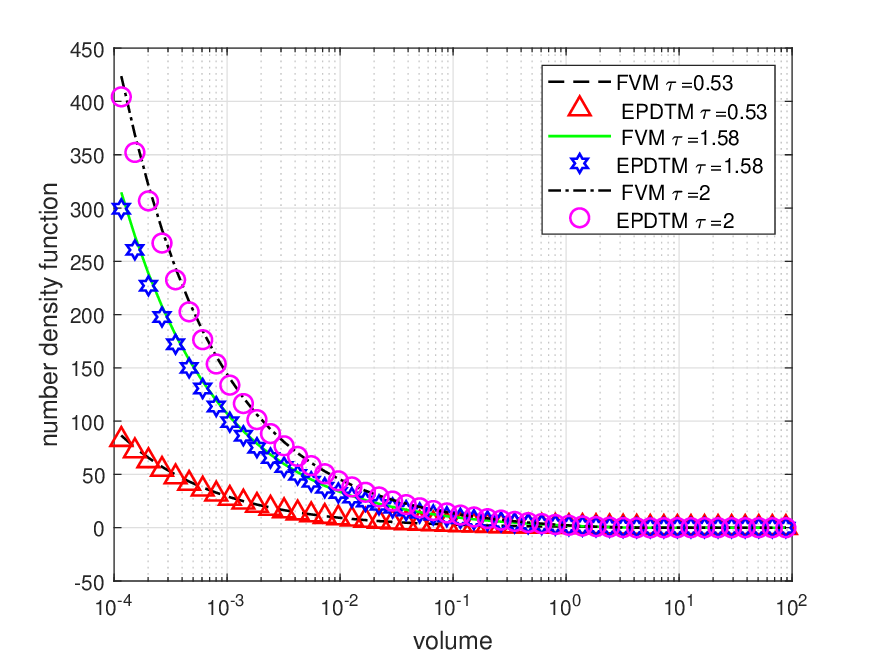}}
\subfigure[FVM and EPDTM zeroth moments ]{\includegraphics[width=0.35\textwidth]{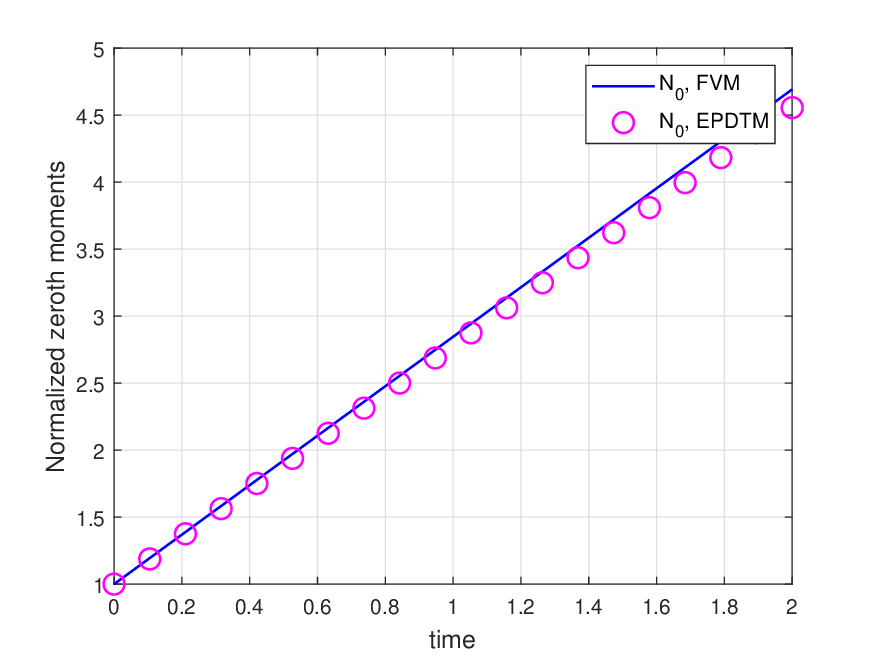}}
\subfigure[FVM and EPDTM first moments ]{\includegraphics[width=0.35\textwidth]{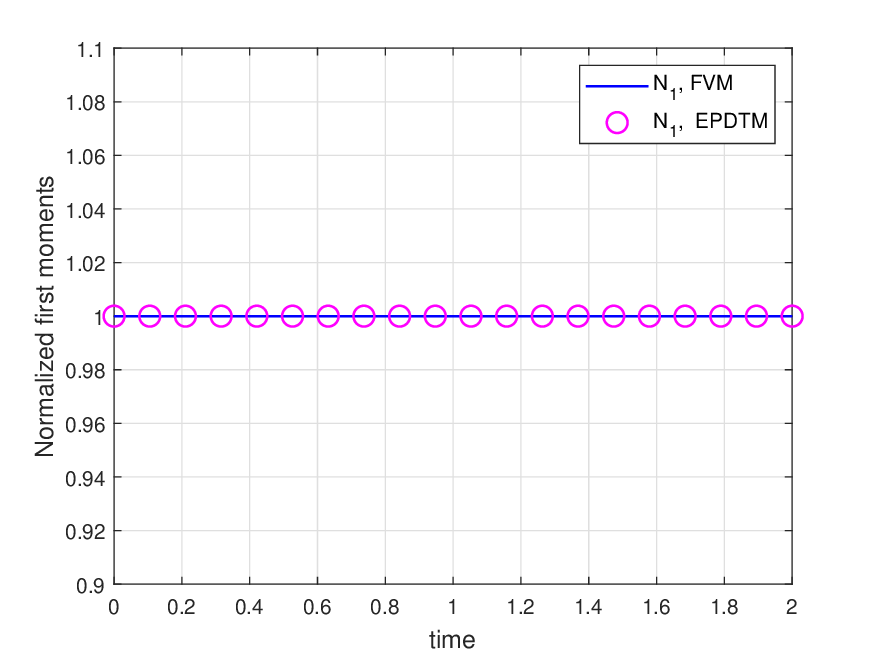}}
\subfigure[FVM and EPDTM second moments ]{\includegraphics[width=0.35\textwidth]{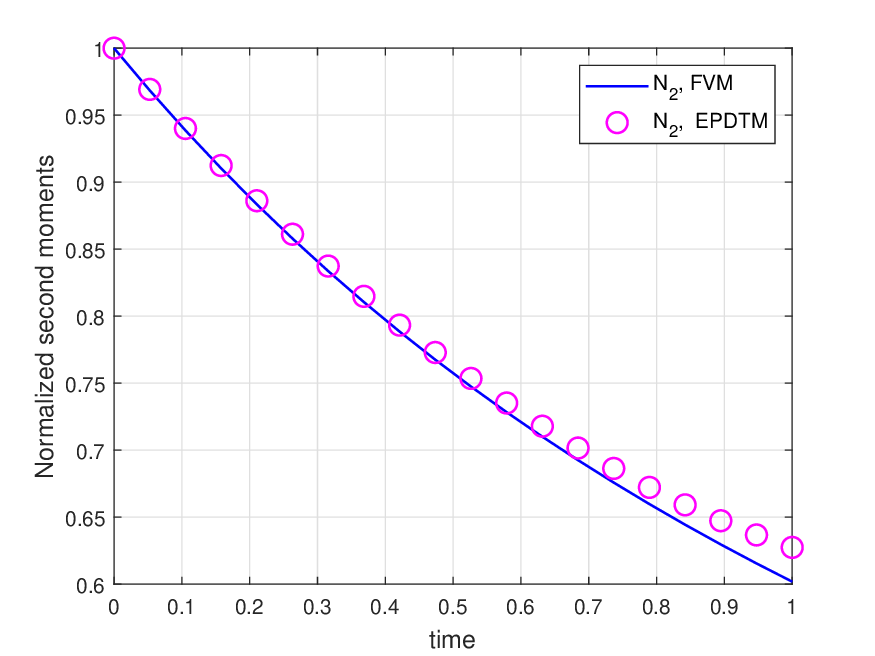}}
\caption{Particle density distribution and moments functions for Test case 6}
\label{fig6}
\end{figure}
Figure \ref{fig6}(a) compares the FVM and EPDTM solutions with 4th-order approximation at various time levels. Interestingly, for all times, the FVM and EPDTM solutions align well with each other. The zeroth and first moments computed using EPDTM also match closely with the moments obtained from the FVM solution, as seen in Figures \ref{fig6}(b) and \ref{fig6}(c). Similar results have been noticed for the second moment but only for shorter time as shown in Figure \ref{fig6}(d). As time progresses, EPDTM starts deviating.

\section{Conclusion and future aims}
This article represented a significant advancement in the field of semi-analytical methods by demonstrating the effectiveness of the EPDTM in solving the collisional breakage population balance equation, even in cases where exact solutions had previously been unattainable. Through a meticulous examination of six diverse test cases, which included variations in initial conditions, collisional kernels, and breakage functions, the study highlighted the precision and superiority of the EPDTM in estimating number density and moments with minimal iterations. Notably, in the first example, the series solution converges exactly to the precise solution found in the existing literature, reinforcing the credibility of the technique. In other examples, where precise solutions were unavailable, the algorithm consistently mirrored the model's predicted behavior with remarkable accuracy, providing solutions in functional form. A comprehensive numerical comparison with the FVM and graphical analysis further affirmed the reliability and robustness of the proposed method. An essential contribution of this work was its thorough analysis of the convergence of iterative series solutions towards the analytical solution and the establishment of error bounds.\\

In the future, to extend the solutions for longer time periods and larger domains, the series of solutions obtained by the EPDTM will be further expanded using powerful Pade' approximants. 

%
%
\bibliography{ref}
\bibliographystyle{ieeetr}
\end{document}